\documentclass[11pt]{article}

\usepackage{dsfont}
\usepackage{geometry}
\usepackage{graphics}
\usepackage{graphicx}
\usepackage{xcolor}
\usepackage{amsmath,amsfonts,amsthm,amssymb,framed,color,bbm}

\newcommand{\amsHackForListFollowsLemma}{\mbox{}\vspace*{-\parsep}\vspace*{-\baselineskip}}

\newtheoremstyle{Susanne}	
{10pt}						
{10pt}						
{\slshape}					
{}							
{\bfseries}					
{}							
{\newline}%
{}%

\theoremstyle{Susanne}
\newtheorem{thm}{Theorem}[section]
\newtheorem{prop}[thm]{Proposition}
\newtheorem{lemma}[thm]{Lemma}

\newtheorem{cor}[thm]{Corollary}

\theoremstyle{definition}

\theoremstyle{remark}
\newtheorem*{rem}{Remark} 	

\numberwithin{equation}{section} 	

\newcommand{\Z}{\mathbb{Z}}
\newcommand{\N}{\mathbb{N}}

\newcommand{\R}{\mathbb{R}}

\renewcommand{\div}{\operatorname{div}}
\DeclareMathOperator{\Cov}{Cov} 	

\newcommand {\T} {\mathbb T} 	
\newcommand{\C}{C_c^{\infty}} 	
\newcommand{\mA}{\mathcal A} 	
\newcommand{\mC}{\mathcal C} 	
\newcommand{\mL}{\mathcal L} 	
\newcommand{\loc}{\textrm{loc}}

\newcommand{\de}{\mathrm{d}}		

\newcommand*{\defeq}{\mathrel{\vcenter{\baselineskip0.5ex \lineskiplimit0pt
                     \hbox{\scriptsize.}\hbox{\scriptsize.}}}%
                     =}			

\begin{document}

\title{Scaling limit and convergence of smoothed covariance for gradient models with non-convex potential}
\author{Susanne Hilger}
\maketitle

\begin{abstract} 
A discrete gradient model for interfaces is studied.
The interaction potential is a non-convex perturbation of the quadratic gradient potential. Based on a representation for the finite volume Gibbs measure obtained via a renormalization group analysis by Adams, Koteck\'{y} and M\"uller in \cite{AKM15} it is proven that the scaling limit is a continuum massless Gaussian free field. From probabilistic point of view, this is a Central Limit Theorem for strongly dependent random fields. Additionally, the convergence of covariances, smoothed on a scale smaller than the system size, is proven.
\end{abstract}



\section{Introduction}

We analyze discrete gradient models which are effective models for random interfaces. Let $\Lambda \subset \Z^d$ be a finite set. To each configuration $\varphi \in \R^{\Lambda}$ ($\varphi(x)$ can be interpreted as the height of the interface at site $x$) an energy $H_{\Lambda}(\varphi)$ is assigned, where the Hamiltonian is assumed to be of gradient type,
$$
H_{\Lambda}(\varphi) = \sum_{x \in \Lambda} \sum_{i = 1}^d W(\nabla_i \varphi (x))
\quad \text{with} \quad
\nabla_i \varphi (x) = \varphi(x+e_i) - \varphi(x).
$$
We consider tilted boundary conditions, i.e., for $u \in \R^d$
$$
\psi_u(x) =  x \cdot u 
\quad \text{for} \quad
 x \in \partial \Lambda= \lbrace z \in \Z^d \setminus \Lambda: \vert z-x \vert = 1 \text{ for some } x \in \Lambda \rbrace.
$$
The finite volume Gibbs measure for inverse temperature $\beta$ is given by
$$
\gamma_{\Lambda}^{\psi_u} (\de \varphi) = \frac{1}{Z_{\Lambda}^{\psi_u}} \,\, e^{-\beta H_{\Lambda}(\varphi)} \prod_{x \in \Lambda} \de \varphi(x) \prod_{x \in \partial \Lambda} \delta_{\psi_u(x)}(\de \varphi (x))
$$
where $Z_{\Lambda}^{\psi_u}$ is the partition function which normalizes the measure.
Later it will be more convenient to work 
with periodic boundary condition rather than Dirichlet boundary condition because the problem then remains translation invariant.
Imposing a tilt $u$ corresponds to working with functions such that $x \mapsto \varphi - x \cdot u $ is periodic. 
This can be reduced to the study of periodic functions by replacing the expression  $W(\nabla_i \varphi(x))$ in the Hamiltonian
by $W(\nabla_i \varphi(x) + u_i)$, see \eqref{shifted potential} below.
For the equivalence of various ways of imposing a tilt, at least on the level of the thermodynamic limit of the free energy,
see \cite{KL14}.

~\\
In the case of strictly convex, symmetric $W$ a lot is known: The infinite volume gradient Gibbs measure exists and is uniquely determined by the tilt \cite{FS97}. The long distance behaviour is described by the Gaussian free field (see \cite{NS97} and \cite{GOS01}) and the decay of the covariance is polynomial as in the massless Gaussian case \cite{DD05}.
Moreover the surface tension is strictly convex \cite{DGI00}.
A nice overview of these results and the used techniques can be found in \cite{Fun05} or \cite{Vel06}.

~\\
Much less is known for models with non-convex potentials. At moderate temperature and zero tilt Biskup and Koteck\'{y} showed in \cite{BK07} the existence of two ergodic infinite volume Gibbs measures for a particular non-convex potential, a mixture of two centered Gaussians.
For this potential it can nevertheless be shown that both gradient Gibbs measures scale to a Gaussian free field \cite{BS11}.

The high temperature regime of potentials of the form
$$
W(\eta) = W_0 (\eta) + g_0 (\eta),
\quad W_0 \text{ strictly convex },
\quad \sqrt{\beta} \Vert g \Vert_{L^1} \text{ small}
$$
is analyzed in \cite{CDM09} and \cite{CD12}. The authors prove strict convexity of the surface tension, uniqueness of the ergodic gradient Gibbs measure, scaling to the Gaussian free field and polynomial decay of the covariance for any tilt. Note that smallness of  $\sqrt{\beta} \Vert g \Vert_{L^1}$ still allows for non-convex $W$.

For low temperatures a finite range decomposition of the Gaussian measure and renormalization group techniques in the spirit of \cite{Bry09} can be used to get first results for potentials which are small non-convex perturbations of the quadratic potential, i.e.,
$$
W(\eta) = \frac{1}{2} \eta^2 + V(\eta),
\quad V \text{ small, but $W$ nonconvex}. 
$$
In \cite{AKM15} a representation for the finite volume Gibbs measure is obtained and is applied there to show strict convexity of the surface tension for small tilt $u$.

The objective of this note is to use the results of \cite{AKM15} to prove a Central Limit Theorem for these models and to show that their behaviour at long distances is governed by a suitable Gaussian free field. On a slightly finer scale we also prove convergence of the covariances of this model.

~\\
Note that in estimates a constant is always denoted by $C$ but can change from line to line.


\section{Setting and Results}

\subsection{Setting}

Let the potential $W: \R \rightarrow \R$ be a perturbation of the quadratic potential,
\begin{align}
W(\eta) = \frac{1}{2} \eta^2 + V(\eta), \quad V: \R \rightarrow \R. \label{potential}
\end{align}

Following \cite{FS97}, we enforce tilted boundary conditions and simultaneously shift invariance by considering fields on the torus and a shifted potential $W(\cdot + u_i)$:
For $L > 0$ a large fixed odd integer consider the discrete torus $\mathbb{T}^d_N = \left(\Z/L^N\Z\right)^d$ of side length $L^N$. This torus can
be represented by the cube $\Lambda_N$ of length $L^N$,
\begin{align}
\Lambda_N \defeq \left\{ x \in \Z^d: |x|_{\infty} \leq \frac{1}{2}(L^N-1) \right\},
\end{align}
equipped with the metric
$d_{\text {per}}(x,y) \defeq  |x-y|_{\text {per}}  \defeq \inf \left\{ |x-y+k|_{\infty}: k \in \left( L^N\Z \right)^d \right\}$ where $|x|_{\infty}\defeq\max_{i=1, \ldots, d}|x_i|$.

Due to the gradient type of the Hamiltonian we restrict to fields with mean value zero,
\begin{align}
\chi_N \defeq \{ \varphi : \Lambda_N \rightarrow \R, ~ \sum_{x \in \Lambda_N} \varphi(x) = 0 \},
\end{align}
equipped with the scalar product $(\varphi, \psi) = \sum_{x \in \Lambda_N} \varphi(x) \psi(x)$.

~\\
By \eqref{potential} the shifted Hamiltonian can be written as
\begin{align}
H_N(\varphi)
= \mathcal E_N(\varphi) + \frac{1}{2} L^{Nd} \left| u \right|^2 + \sum_{x\in \Lambda_N} \sum_{i=1}^d V(\nabla_i \varphi(x) + u_i),
\label{shifted potential}
\end{align}
where
$$
\mathcal E_N(\varphi)
\defeq \frac{1}{2} (\mA^0 \varphi, \varphi)
= \frac{1}{2} \sum_{x\in \Lambda_N} \sum_{i=1}^d \left( \nabla_i \varphi(x) \right)^2,
\quad \mA^0 = \sum_{i,j} \delta_{ij} \nabla_j^* \nabla_i
$$
with the adjoint difference operator $\nabla_i^* \varphi (x) = \varphi(x-e_i)- \varphi(x)$.
The finite volume Gibbs measure is
\begin{align}
\nu_{H_N} (\de \varphi) = \frac{1}{Z_{H_N}} \, e^{- \beta H_N (\varphi)} \, \de \lambda_N (\varphi)
\end{align}
where $Z_{H_N}= \int_{\chi_N} e^{- \beta H_N(\varphi)} \, \de \lambda_N (\varphi)$ is the partition function and $\de\lambda_N$ denotes the $(L^{Nd}-1)$-dimensional Hausdorff measure  on $\chi_N$.

\begin{rem}
To simplify the notation we set $\beta = 1$ in the following.
Note further that we do not make explicit the dependence on $u$ in the notation since it plays no special role here.
\end{rem}


\subsection{Results}

For computing the scaling limit we define, for $f \in \C(\T^d;\R^d)$,
\begin{align}
f^N(x) \defeq L^{-N\frac{d}{2}} f(L^{-N}x) \label{scaled f}
\end{align}
and introduce the slowly varying scaled field
\begin{align}
\nabla \varphi (f^N) \defeq (\nabla\varphi, f^N) \text{ for } x \in \Lambda_N.
\end{align}
We describe the distribution of this random vector by the Laplace transform
\begin{align}
\int e^{- (\nabla\varphi, f^N)} \,\nu_{H_N} (\de \varphi) \quad \textrm{for } f \in \R^{\Lambda_N}.
\end{align}

We are interested in the existence of some limiting distribution $(\nabla\varphi)^{\text{scale}}$ which is in general a random field on distributions.
Then the limiting distribution, if it exists, is the joint distribution $\nu^{\text{scale}}$ for generalized gradient random fields $\left((\nabla\varphi)^{\text{scale}}(f) \right)_{f \in C_c^{\infty}(\mathbb T^d)}$,\linebreak $\T^d~=~\R^d/\Z^d$, such that
$$
\mathbb E_{\nu^{\text{scale}}} \left( e^{ -\nabla\varphi^{\text{scale}} (f)} \right)
= \lim_{N \rightarrow \infty}
\int_{\chi_N} e^{-(\nabla\varphi, f^N)} \,\nu_{H_N}(\de\varphi)
\quad \textrm{for all }f \in C_c^{\infty}(\mathbb T^d).
$$

We will show that for small initial perturbation of the quadratic potential the measure $\nu_{H_N}$ tends to a continuum Gaussian free field  with a renormalized covariance in the sense of convergence of Laplace transforms of the measures.

For stating the smallness condition on $V$ we introduce the second order Taylor remainder $U$ of $V$, 
$$U(s,t) = V(s+t) - V(t) - V'(t)s,$$
and set, for $z \in \R^d$,
\begin{align}
\mathcal K (z) = e^{- \sum_{i=1}^d U(z_i, u_i)} - 1 .
\end{align}
{We define, for  a multiindex $\alpha = (\alpha_1, \ldots, \alpha_d)$ with $ \alpha_i \in \N$, the length 
$|\alpha| = \sum_{i=1}^d \alpha_i$  and the operator 
$\partial ^{\alpha} = \prod_{i=1}^d \partial_i^{\alpha_i}$ with $\partial^0_i \defeq 1$. For  $\zeta >0$ and $r_0 \in \N$
we define  the  norm
\begin{align}
\left\| \mathcal K \right\|_{\zeta}
= \sup_{z \in \R^d} \sum_{|\alpha| \leq r_0} \zeta^{|\alpha|} \left| \partial_z^{\alpha} \mathcal K(z) \right| e^{- \zeta^2|z|^2}.
\end{align}

\begin{thm}\label{Scaling limit}
There is $\zeta > 0$, $r_0 \in \N$ and $\rho > 0$ such that for all $V$ with $\Vert\mathcal K \Vert_{\zeta} \leq \rho$ there is $(\bar{q}_{ij})_{i,j \in \lbrace 1, \ldots, d \rbrace}$ satisfying for any $f \in \C (\T^d; \R^d)$ on a subsequence
$$
\int e^{-(\nabla \varphi, f^N)} \nu_{H_N}(\de \varphi)
\,\rightarrow\, e^{\frac{1}{2}(\div f, \,\mC \div f)_{L^2(\T^d)}}
$$
as $N$ tends to infinity, where the right hand side is the Laplace transform of the continuum Gaussian free field on $\T^d$ with covariance $\mC = (\mA)^{-1}$,
$$
\mA \defeq - \sum_{i,j = 1}^d \left(\delta_{ij} + \bar{q}_{ij}\right) \partial_j \partial_i.
$$
\end{thm}

\begin{rem}
This is a Central Limit Theorem for strongly dependent random fields. Indeed, let $f$ be an approximation of the characteristic function of $Q(0) = \left( -\frac{1}{2}; \frac{1}{2}  \right)^d $. By the above Theorem the limiting distribution of
$$
\phi^N \defeq L^{-N\frac{d}{2}} \sum_{x \in \Lambda_N} \nabla \varphi (x)
$$
is Gaussian. This also explains the choice of the scaling factor used here which is the typical one in Central Limit Theorems. The classical Central Limit Theorem cannot be applied due to the long-range gradient-gradient correlations of the measure.
\end{rem}

~\\
Theorem \ref{Scaling limit} captures the limiting behaviour if we average over the whole system of scale $L^{dN}$ and rescale. 
If we are not interested in the full distribution but only in covariances we can also show convergence
to the covariances of a Gaussian measure if we only average over $L^{\alpha d N}$ many points, $\alpha < 1$.
To state this result precisely
fix $a,b \in  \R^d$ and, for $z = a,b$, let $J_z \in \C\left(Q(z);\R^d\right)$, where $Q(z) = z + \left(-\frac{1}{2},\frac{1}{2} \right)^d$.
Define for $0 < \alpha \leq 1$
\begin{align}
J^N_z (x)
\defeq L^{-\alpha N \frac{d}{2}} J_z \left(\frac{x}{L^{\alpha N}}\right)
\quad \textrm{for } x \in \Lambda_N.
\label{scaled J}
\end{align}

\begin{thm}\label{Smoothed covariance}
For $\zeta$, $r_0$, $\rho$, $V$ and $\bar{q}$ as in Theorem \ref{Scaling limit} there is $\alpha_0 < 1$ such that for any $\alpha \in \left[ \alpha_0;1 \right]$ and for all $J_z \in \C \left((Q(z); \R^d\right)$ on a subsequence
$$
\Cov _{\nu_{H_N}}\left( (\nabla \varphi, J_a^N), (\nabla \varphi J_b^N) \right)
\rightarrow
\left( \div^* J_a, \mC \div^* J_b \right)_{L^2(\R^d)}
$$
as $N$ tends to infinity, where the right hand side is the covariance of the continuum Gaussian free field
on $\R^d$, i.e., $\mC$ is the inverse of the operator $\sum_{i,j=1}^d \left(\delta_{ij} + \bar q_{ij}\right) \partial_j \partial_i$
defined on functions on $\R^d$.
\end{thm}


\section{Outline of the Proofs}

The proofs of Theorem \ref{Scaling limit} and \ref{Smoothed covariance} 
rely on a representation for the finite volume Gibbs measure constructed in \cite{AKM15}. There, the measure $\nu_{H_N}$ is written as perturbation of a Gaussian measure $\mu^q$, i.e., $\nu_{H_N} = F_0^q(\varphi) \mu^q$. In \cite{AKM13} $\mu^q$ is decomposed into Gaussian measures $\mu^q_1, \ldots, \mu^q_{N+1}$ with a suitable finite range property such that $\mu^q = \mu_1^q \ast \ldots \mu_{N+1}^q$. This yields
\begin{align}
\int F_0^q(\varphi) \mu^q (\de \varphi)
&= \int F_0^q(\varphi) \, \mu_1^q \ast \ldots \ast \mu_{N+1}^q (\de \varphi)
= \int F^q_1(\varphi) \, \mu_2^q \ast \ldots \ast \mu_{N+1}^q (\de \varphi)
\nonumber \\
&= \ldots
= \int F^q_N(\varphi) \, \mu_{N+1}^q (\de \varphi). \label{iterative integration}
\end{align}
It can be shown that  the effective interactions $F^q_k$ are expressed as the composition of a 'relevant' term $e^{-H_k^q(\varphi)}$, 
where $H_k^q(\varphi)$ is quadratic and local in $\nabla \varphi$,  and an 'irrelevant term' $K_k^q(\varphi)$ such that the
map $K_k \mapsto K_{k+1}$ is a contraction in suitable norms. From this contraction property one can deduce a Stable
Manifold Theorem for the evolution of the $(H_k^q, K_k^q)$ variables.
This gives the possibility to choose $q$ such that the initial perturbation $F^q_0$ lies on the stable manifold.
Hence after $N$ steps a nice representation, which corresponds to $H^q_N=0$, is obtained.

~\\
In the following we first give a precise statement of corresponding result in \cite{AKM15} and collect useful consequences. Then we give a sketch of the proof of Theorem \ref{Scaling limit} and Theorem~\ref{Smoothed covariance}.
In the next Section we provide the detailed proofs.


\subsection{Representation for $\nu_{H_N}$}

To state the result in \cite{AKM15} we have to introduce some objects. First of all we rewrite the Gibbs measure $\nu_{H_N}(\de \varphi)$ as perturbation of a Gaussian measure $\mu^{q}$, $q \in \R^{d \times d}_{\text{sym}}$, with covariance $\mC^q = (\mA^q)^{-1}$, where
\begin{align}
\mA^q = \sum_{i,j} (\delta_{ij} + q_{ij}) \nabla_j^* \nabla_i.
\end{align}
For this recall that $U$ is the second order Taylor remainder of $V$, i.e.,
\begin{align}
U(s,t) = V(s+t) - V(t) - V'(t)s \label{U}
\end{align}
and insert artificially the so called fine-tuning parameters $\frac{1}{2}(q \nabla \varphi, \nabla \varphi)$ and $\lambda^q \in \R$
to get
\begin{align}
e^{-H_N(\varphi)} \,\de \lambda_N(\varphi)
= \kappa \, F(\varphi) \, \mu^q(\de \varphi)
\label{trafo real Z}
\end{align}
where $\kappa = e^{-\frac{1}{2} L^{Nd} \left| u \right|^2} e^{-L^{Nd} \sum_i V(u_i)} e^{-\lambda^q L^{Nd}} Z_N^q$ 
and
\begin{align}
F(\varphi) = e^{\frac{1}{2} \left( q \nabla \varphi, \nabla \varphi \right) + \lambda^q L^{Nd}} e^{- \sum_x \sum_i U \left( \nabla_i \varphi (x), u_i \right)}.
\end{align}
Here $Z_N^q$ is the partition function of the measure $\mu^q$.
Then
\begin{align}
\nu_{H_N} (\de \varphi)
= F(\varphi) \mu^q(\de \varphi) \frac{1}{\int F(\varphi) \mu^q(\de \varphi)}.
\label{trafo Z}
\end{align}

Furthermore we need the norm on the irrelevant part used in the last integration step in~\eqref{iterative integration}. In the following several constants will appear which are needed for the construction in \cite{AKM15} but  they will not be explained or motivated here.

First we define a norm on fields $\varphi \in \R^{\Lambda_N}$ by
\begin{align}
\left| \varphi \right|_{N,\Lambda_N} = \underset{1\leq s\leq3}{\text{max}} ~ \underset{x \in \Lambda_N}{\sup} \frac{1}{h} ~ L^{N\left( \frac{d-2}{2} + s \right)} \left| \nabla^s \varphi (x) \right|
\end{align}
where $\left| \nabla^s \varphi (x) \right|^2 = \sum_{|\alpha|=s} \left| \nabla^{\alpha} \varphi(x) \right|^2$, 
$\alpha$  is a multiindex and $h > 0$.
We introduce the quantities
\begin{align}
G_{N,x}(\varphi)
&= \frac{1}{h^2} \left(  \left| \nabla \varphi(x) \right|^2 + L^{2N}  \left| \nabla^2 \varphi(x) \right|^2 + L^{4N}  \left| \nabla^3 \varphi(x) \right|^2 \right) \label{norm G}
\\
\mathrm{and} \quad
g_{N,x}(\varphi)
&= \frac{1}{h^2}  \sum_{s=2}^4 L^{(2s-2)N} \sup_{y \in \Lambda_N} \left| \nabla^s \varphi(y) \right|^2 \label{norm g}
\end{align}
which are used to define the so-called large field regulator
\begin{align}
w_N^{\Lambda_N} (\varphi) = e^{\sum_{x \in \Lambda_N} \omega (2^d g_{N,x}(\varphi) + G_{N,x}(\varphi))}, \quad \omega \in \R.
\label{large field regulator}
\end{align}

Next we determine a seminorm which controls the Taylor remainder of a function $F$ on fields,
\begin{align}
\left| F(\varphi) \right|^{N,\Lambda_N}
= \sum_{s=0}^{r_0} \frac{1}{s!} \sup_{|\bar{\varphi}|_{N,\Lambda_N} \leq 1} \left|D^s F(\varphi)(\bar{\varphi}, \ldots, \bar{\varphi})\right|
\end{align}
for some $r_0 >0$.
Finally we define the norm
\begin{align}
\Vert F  \Vert_{N, \Lambda_N}
= \sup_{\varphi} \left| F(\varphi) \right|^{N,\Lambda_N} w_N^{-\Lambda_N} (\varphi).
\end{align}
Note that, in comparison to \cite{AKM15}, we skip any dependencies of maps on subsets $X \subset \Lambda_N$ since we do not need it here.

~\\
Let
$$
\mu^q = \mu^q_1 \ast \ldots \ast \mu^q_{N+1}
$$
be a decomposition of $\mu^q$ into Gaussian measures with range on increasing blocks, see \cite{AKM13} for an exact definition and for existence of such a decomposition whenever $q$ is small enough, i.e., $\Vert q \Vert \leq \frac{1}{2}$ where the norm $\left\| q \right\|$ is the operator norm of $q$ viewed as operator on $\R^d$ equipped with the metric $|q|_{l_2} = \left( \sum_i |q_i|^2 \right)^{\frac{1}{2}}$.

~\\
In the following Proposition we use the notation $$F \ast \mu (\xi) = \int F(\varphi + \xi) \mu(\de \varphi).$$


\begin{prop}\label{resultAKM}
There exist positive constants $\rho$, $\rho_1 \leq \frac{1}{2}$, $\zeta$ and $\eta \in (0,1)$, $\eta$ independent on $N$, such that for suitable chosen constants $L$, $h$, $\omega$ and $r_0$ and for any $\mathcal K$ with $\Vert \mathcal K \Vert_{\zeta} \leq \rho$ there is a parameter $q = q(\mathcal K,N)$ with $\Vert q \Vert \leq \rho_1$ satisfying
\begin{align} F \ast \mu^q (\xi)
= 1 + \int K_N (\varphi + \xi) \,\mu_{N+1}^q (\de \varphi) \label{representation AKM}
\end{align}
such that
\begin{align}
\Vert K_N \Vert_{N, \Lambda_N} \leq C \, \eta^N. \label{decay norm AKM}
\end{align}
\end{prop}

A choice for $L$, $h$, $\omega$ and $r_0$ is made in \cite{AKM15}, Proposition 4.6. The existence of the constants $\rho$, $\rho_1$ and $\zeta$ and of the parameter $q$ and the formula \eqref{representation AKM} can be found in \cite{AKM15}, Theorem 4.9. The exponential decay of the norm \eqref{decay norm AKM} is a consequence of Proposition 8.1 and of the construction of the corresponding Banachspace in Subsection 4.5, both in \cite{AKM15}.

~\\
From the results in \cite{AKM15} 
one can also deduce the following estimate on maps $F$ on fields.
\begin{lemma} \label{derivative estimate}
For any $s \leq r_0$ it holds
\begin{align*}
&\left| \int D^s F (\varphi + \xi)(\xi_1,\ldots,\xi_s) \mu^q_{N+1} (\de \varphi) \right|
\\
& \qquad \leq
C \left| \xi_1 \right|_{N,\Lambda_N} \ldots \left| \xi_s \right|_{N,\Lambda_N}
\left\| F \right\|_{N,\Lambda_N}
\left( w_N^{\Lambda_N} (\xi) \right)^2.
\end{align*}
\end{lemma}

\begin{proof}
By looking carefully at the definition of the norm $\Vert \cdot \Vert_{N, \Lambda_N}$ 
one  can easily obtain the following estimate:
\begin{align*}
&\left| \int D^s F (\varphi + \xi)(\xi_1,\ldots,\xi_s) \mu^q_{N+1} (\de \varphi) \right|
\\
& \qquad \leq
C \left| \xi_1 \right|_{N,\Lambda_N} \ldots \left| \xi_s \right|_{N,\Lambda_N}
\left\| F \right\|_{N,\Lambda_N}
\int  w_N^{\Lambda_N}(\varphi + \xi)\, \mu^q_{N+1} (\de \varphi).
\end{align*}
An adjustment of the proof of Lemma 5.2 in \cite{AKM15} (there the case $k=N$ is excluded) gives
	$$
	\int w_N^{\Lambda_N}(\varphi + \xi) \, \mu^q_{N+1}(\de\varphi)
	\leq 
	C \left( w_N^{\Lambda_N} (\xi) \right)^2.
	$$
In fact the adjustment is a huge simplification since we do not have to deal with the boundary terms which the authors of \cite{AKM15} have to for the scales $k<N$.
\end{proof}


\subsection{Sketch of the proofs}
In order to apply Proposition \ref{resultAKM} for the computation of the scaling limit and the smoothed covariance we do the following \emph{key calculation}:
By completing the square and linear transformation we get for a Gaussian measure $\mu_C$
\begin{align}
\int e^{-(\varphi,f)} F(\varphi) \mu_C(\de\varphi)
&= \frac{1}{Z} \int e^{-(\varphi,f)} e^{-\frac{1}{2}(\varphi, C^{-1} \varphi)}F(\varphi) \,\de\varphi \nonumber
\\
&= e^{\frac{1}{2}(f,Cf)} \frac{1}{Z} \int e^{-\frac{1}{2}(\varphi + Cf, C^{-1} (\varphi+ C f))}F(\varphi) \,\de\varphi \nonumber
\\
&= e^{\frac{1}{2}(f,Cf)} \int F(\varphi - Cf) \mu_C(\de\varphi). \label{key calculation}
\end{align}


\begin{proof}[Proof of Theorem \ref{Scaling limit}, main ideas]
We compute, using \eqref{trafo Z} and \eqref{key calculation} and denoting\linebreak $g^N~=~\sum_l\nabla_l^* f^N_l$,
\begin{align*}
\int e^{-(\nabla \varphi, f^N)} \nu_{H_N}(\de \varphi)
&= \frac{1}{\int F \, \de \mu^q} \int e^{-(\varphi, g^N)} F(\varphi) \mu^q(\de\varphi)
\\
&= e^{\frac{1}{2}(g^N,\, \mC^q g^N)} \,
\frac{1}{\int F \, \de \mu^q} \, \int F(\varphi - \mC^q g^N) \mu^q(\de \varphi)
\\
&= e^{\frac{1}{2}(g^N,\, \mC^q g^N)} \,
\frac{F \ast \mu^q (-\mC^q \, g^N)}{F \ast \mu^q (0)}.
\end{align*}

Now we apply the representation \eqref{representation AKM} in Proposition \ref{resultAKM}, Lemma \ref{derivative estimate} and the bound \eqref{decay norm AKM} in Proposition \eqref{resultAKM} to see that
\begin{align*}
\left| F \ast \mu^q (0) - 1 \right|
= 	\left|
			\int K_N (\varphi) \mu^q_{N+1} (\de\varphi)
		\right|
\leq C \Vert K_N  \Vert_{N, \Lambda_N} w_N^{\Lambda_N}(0)^2
\leq C \eta^N w_N^{\Lambda_N}(0)^2.
\end{align*}
By the same reasoning,
\begin{align*}
\left| F \ast \mu^q (-\mC^q \, g^N) - 1 \right|
 \leq C \eta^N w_N^{\Lambda_N}(-\mC^q \,g^N)^2.
\end{align*}

In Subsection \ref{errorterms}, Lemma \ref{last scale bound} we show that $w_N^{\Lambda_N}(0)$ and $w_N^{\Lambda_N}(-\mC^q \,g^N)$ can be bounded uniformly in $N$ such that
$$
\frac{ F \ast \mu^q (-\mC^q \,g^N)}{ F \ast \mu^q (0)}\rightarrow 1
\quad \text{as } N \rightarrow \infty.
$$
The convergence
$$
e^{\frac{1}{2}(g^N,\,\mC^q g^N)}
\rightarrow e^{\frac{1}{2}(\div f,\,\mC \div f)_{L^2}}
\quad \text{as } N \rightarrow \infty
$$
is proved in detail in Subsection \ref{convergence}, see Proposition \ref{convergence torus}.
\end{proof}


For the proof of Theorem \ref{Smoothed covariance} note that we can compute the covariance by taking derivatives of the logarithm of a generating functional.
For a measure $\nu = \frac{1}{Z} \, e^{-H} \, \de \lambda$ and random variables $X$ and $Y$ it holds
\begin{align}
\Cov_{\nu} (X,Y)
= \partial_s \partial_t \Big \vert _{s=t=0} \ln \int e^{-(sX+tY)} e^{-H} \, \de \lambda.
\label{cov as der}
\end{align}

\begin{proof}[Proof of Theorem \ref{Smoothed covariance}, main ideas]

We use the representation of the covariance \eqref{cov as der} as follows
$$
\textrm{Cov}_{\nu_{H_N}} ((\nabla \varphi, J_a^N),(\nabla \varphi, J_b^N))
= \partial_s \partial_t \Big \vert _{s=t=0} \ln
\int e^{-(sJ_a^N + tJ_b^N, \nabla \varphi)} e^{-H_N(\varphi)} \,\de \lambda_N(\varphi).
$$
As in the proof of Theorem \ref{Scaling limit} we compute, using \eqref{trafo real Z},  \eqref{representation AKM}
 and \eqref{key calculation} and denoting\linebreak $g^N_z=\sum_l \nabla_l^* J_z^N$ for $z=a,b$,
\begin{align}
& \Cov_{\nu_{H_N}} ((\nabla \varphi, J_a^N),(\nabla \varphi, J_b^N))
\nonumber \\
&= \partial_s \partial_t \Big \vert _{s=t=0} \ln
	\left[ \kappa \,
		e^{\frac{1}{2}\left(s g_a^N + t g_b^N, \,\mC^q \left(s g_a^N + t g_b^N\right)\right)}
		  F \ast \mu^q (- s \mC^q\, g_a^N - t \mC^q\, g_b^N)
	\right]
\nonumber \\
&= \left( g_a^N, \mC^q\, g_b^N \right)
\nonumber \\
& \qquad +	\frac{1}{F \ast \mu^q (0)}
\int D^2 K_N (\varphi)(-\mC^q\, g_a^N, -\mC^q\, g_b^N) \mu^q_{N+1}(\de\varphi)
\label{cov extended}\\
& \qquad -	\frac{1}{F \ast \mu^q (0)^2}
\int D K_N (\varphi)(-\mC^q\, g_a^N) \mu^q_{N+1}(\de\varphi)
\int D K_N (\varphi)(-\mC^q\, g_b^N) \mu^q_{N+1}(\de\varphi).
\nonumber
\end{align}

As in the proof of Theorem \ref{Scaling limit} it holds by the use of Proposition \ref{resultAKM}
$$
\left| F \ast \mu^q (0) - 1 \right|
\leq C \eta^N w_N^{\Lambda_N}(0)^2.
$$
and thus the denominators in \eqref{cov extended} tend to $1$ as $N$ tends to infinity.
Employing Lemma~\ref{derivative estimate} and \eqref{decay norm AKM} in Proposition \ref{resultAKM} we further get
\begin{align*}
\left|
	\int D^2 K_N (\varphi)(-\mC^q\, g_a^N, -\mC^q\, g_b^N) \mu^q_{N+1}(\de\varphi)
\right|
\leq C \eta^N
\left| \mC^q\, g_a^N \right|_{N, \Lambda_N}
\left| \mC^q\, g_b^N \right|_{N, \Lambda_N}
\end{align*}
and for $z = a,b$
\begin{align*}
\left|
	\int D K_N (\varphi)(-\mC^q\, g_z^N) \mu^q_{N+1}(\de\varphi)
\right|
& \leq C \eta^N
\left| \mC^q\, g_z^N\right|_{N, \Lambda_N}.
\end{align*}

Hence it remains to show
\begin{align}  \label{eq:error_terms_2}
\eta^N
\left| \mC^q\, g_a^N \right|_{N, \Lambda_N}
\left| \mC^q\, g_b^N \right|_{N, \Lambda_N}
\rightarrow 0 \quad \textrm{as } N \rightarrow \infty.
\end{align}

The exponential decay $\eta^N$ allows us to consider values $\alpha \leq 1$. 
In Lemma \ref{scope alpha} we show that there exist  a decreasing function $\tau$ satisfying $\tau (1) = 1$ such that
$$
\left| \mC^q \, g_z^N \right|_{N, \Lambda_N}
\leq C \, \tau(\alpha)^N.
$$
This gives the possibility to choose $\alpha_0 < 1$ as small as possible such that $\eta \tau(\alpha_0)^2 < 1$ and hence $(\eta \tau(\alpha)^2) ^N \rightarrow 0$ as $N$ tends to infinity for all $\alpha \in \left[\alpha_0,1\right]$.

Then it only remains to show
$$
\left( g_a^N, \mC^q \,  g_b^N \right)
\rightarrow
\left( \div J_a, \mC \, \div J_b \right)_{L^2} .
$$
This is done in Subsection \ref{convergence}, Proposition \ref{convergence whole space}.
\end{proof}


\begin{rem}
For higher moments (say of order $s$) one has to choose $\alpha_0$ such that
$$
\eta \tau(\alpha_0)^s < 1
$$
which implies $\alpha_0 > 1- \frac{-\ln \eta}{s(\frac{d}{2} + 4)}$ (see the formula for $\tau(\alpha)$
in Lemma \ref{scope alpha}) so $\alpha_0 \rightarrow 1$ as $s \rightarrow \infty$.
Estimates of correlation functions can be done in more generality by inserting external fields into the partition function and extending the flow of $(H_k^q, K_k^q)$ to these observable variables, compare \cite{BBS14}. This amounts to extending norm estimates in \cite{AKM15} to the case of included variables and computing explicitly the flow of the observables. We plan to pursue this in future work.
\end{rem}

\section{Details of the Proofs}

\subsection{Scaled discrete setting}

For the proof of the convergence of $\mC^q$ to $\mC$ we switch to a scaled setting. Let $0 < \alpha \leq 1$. Define
$$
\mathcal A_{N} = \sum_{i,j} a^N_{ij} \, D_{N,j}^* D_{N,i}
 \quad \text{with } a^N_{ij} = \delta_{ij} + q^N_{ij}
$$
where
$$
D_{N,j} \varphi (x)
= \frac{\varphi(x + L^{-\alpha N} e_j ) - \varphi(x)}{L^{-\alpha N}}
\quad \text{and} \quad
D_{N,j}^* \varphi (x)
= \frac{\varphi(x - L^{-\alpha N} e_j ) - \varphi(x)}{L^{-\alpha N}}.
$$
We will also use
\begin{align}
\div_N^* h(x) = \sum_l D_{N,l}^* h_l(x) \quad \text{for} \quad h: \Lambda_N \rightarrow \R^d.
\label{discrete divergence}
\end{align}
Further define the scaled discrete torus
$$
\Lambda_N' = \Lambda_N/L^{\alpha N}
$$
of spacing $L^{\alpha N}$ with fundamental domain embedded into $\left[ -\frac{L^{(1-\alpha)N}}{2}, \frac{L^{(1-\alpha)N}}{2} \right]^d$. The corresponding torus in continuum is thus $\T^d_{R_N}$ for $R_N = L^{(1-\alpha)N}$.

~\\
For maps $\varphi, \psi: \Lambda_N' \rightarrow \R$ define
\begin{align*}
\left\langle \varphi, \psi \right\rangle_{l^2}
&\defeq \,
L^{-\alpha Nd} \sum_{x \in \Lambda_N'} \varphi(x) \psi(x),
\\
\left\langle \varphi, \psi \right\rangle_{w^{1,2} }
&\defeq  \,
\left\langle \varphi, \psi \right\rangle_{l^2 }
+
L^{-\alpha Nd} \sum_{x \in \Lambda_N' } \sum_{k=1}^d D_{N,k}\varphi(x) D_{N,k}\psi(x).
\end{align*}

Further, let
$$
\chi_N' = \left\{ \varphi: \Lambda'_N \rightarrow \R, \sum_{x \in \Lambda'_N} \varphi(x) = 0 \right\}
$$
which becomes a Hilbert space when endowed with the scalar product $\left\langle \cdot, \cdot \right\rangle_{w^{1,2}}$.

~\\
For convenience of the reader we include a proof of the following standard result.
\begin{prop}\label{proposition discrete tori}
For $g: \Lambda_N' \rightarrow \R$ there is a unique weak solution $u_N \in \chi_N'$ to
$$
\mA_N u_N = D_{N,l}^* g \quad\textrm{in } \Lambda_N'.
$$
This solution satisfies
$$
\left\| u_N \right\|_{w^{1,2}}
\leq C \, L^{2N(1- \alpha)} \left\| g \right\|_{l^2}.
$$
Moreover, there is a constant independent of $N$ such that
$$
\left\| D_N u_N \right\|_{l^{2}}
\leq C \left\| g \right\|_{l^2}.
$$
\end{prop}

\begin{proof}
This works as in the continuum. Define
$$
\mL_N (\varphi,\psi) \defeq
L^{-\alpha Nd} \sum_{x \in \Lambda_N'}
\sum_{i,j} a^N_{ij}  D_{N,i} \varphi (x) D_{N,j} \psi(x)
$$
and
$$
F(\psi) \defeq
L^{-\alpha Nd} \sum_{x \in \Lambda_N'}
g(x) D_{N,l} \psi(x).
$$
In $\chi_N'$ a Poincar\'{e} inequality holds (can be found, e.g., in \cite{BGM04}, Lemma B.2) and thus $\mL_N$ is a continuous coercive bilinear form. Indeed, coercivity follows from
\begin{align}
\mL_N (\psi,\psi)
\geq \lambda_{\min}^N \sum_k \left\| D_{N,k} \psi \right\|_{l^2}^2,
\label{mL}
\end{align}
where we denote by $\lambda^N_{\min}$ the smallest eigenvalue of $a_{ij}^N$.
By the Poincar\'{e} inequality it holds
$$
\left\| \psi \right\|_{w^{1,2}}^2
= \left\| \psi \right\|_{l^2}^2 + \sum_k \left\| D_{N,k} \psi \right\|_{l^2}^2
\leq (1 + C(\Lambda_N')) \sum_k \left\| D_{N,k} \psi \right\|_{l^2}^2.
$$
where $C(\Lambda_N') = C_d \, L^{2N(1-\alpha)}$.
Hence
$$
\mL_N (\psi,\psi)
\geq \frac{\lambda_{\min}^N}{1+C(\Lambda_N')} \left\| \psi \right\|_{w^{1,2}}^2.
$$
Further, $F$ is an element of the dual space by the estimate
\begin{align}
\left|F(\psi)\right|
\leq \left\| g \right\|_{l^2} \sum_k \left\| D_{N,k} \psi \right\|_{l^2}
\leq \left\| g \right\|_{l^2} \left\| \psi \right\|_{w^{1,2}}.
\label{F}
\end{align}
Thus the Lax-Milgram Theorem provides a unique solution $u_N \in \chi_N'$ of $\mL_N (u_N,\psi) = F(\psi)$ for all $\psi \in \chi_N'$
together with the estimate
$$
\left\| u_N \right\|_{w^{1,2}}
\leq \frac{1+C(\Lambda_N')}{\lambda_{\min}^N} \left\| F \right\|_{(\chi_N')^*}
\leq \frac{1+C}{\lambda_{\min}^N} L^{2(1-\alpha)N} \left\| g \right\|_{l^2}.
$$
For the estimate which is independent on $N$ we use \eqref{F} with $\psi = u_N$ and \eqref{mL} to obtain
\begin{align*}
\left\| D_N u_N \right\|_{l^{2}}^2
&= \sum_k \left\| D_{N,k} u_N \right\|_{l^{2}}^2
\leq \frac{1}{\lambda_{\min}^N} \left| \mL_N (u_N, u_N) \right|
= \frac{1}{\lambda_{\min}^N}\left|F(u_N)\right|
\\
&\leq \frac{1}{\lambda_{\min}^N}\left\| g \right\|_{l^2} \sum_k \left\| D_{N,k} u_N \right\|_{l^2}.
\end{align*}
Finally note that by the convergence of $q^N$ to $\bar{q}$ (see Lemma \ref{convergence of q}) $\lambda_{\min}^N$ can be bounded from below by $\lambda_{\min} - \epsilon > 0$, $\lambda_{\min}$ being the minimal eigenvalue of $(a_{ij})$.
\end{proof}


\begin{cor}\label{discrete regularity}
Let $u_N$ be the solution in Proposition \ref{proposition discrete tori} and $g \in \C$. Then
$$
\sup_{x \in \Lambda_N'} \left| u_N(x) \right|
\leq C L^{2(1-\alpha)N}
\left\| g \right\|_{C^{d}}.
$$
\end{cor}

\begin{proof}

By discrete differentiation of the strong  form of the equation one obtains \linebreak
$
\mA_N D_N^{s} u_N = D_{N,l}^* D_N^{s}g
$
for any multiindex $s$. We apply Proposition \ref{proposition discrete tori} to get
$$
\left\| D_N^{s} u_N \right\|_{w^{1,2}} \leq C L^{2N(1-\alpha)} \left\| D_N^{s} g \right\|_{l^{2}}.
$$ 
Now we use the discrete 
Sobolev embedding (see, e.g., \cite{BGM04}, Lemma B.1)
to obtain
$$
\sup_x \left| u_N(x) \right|~\leq~C_d \sum_{|s| \leq d} \left\| D_N^{s} u_N \right\|_{l^2}.
$$
This gives the desired estimate.
\end{proof}


We now transform the terms of interest from the discrete unscaled into the discrete scaled setting.

\begin{lemma}\label{Trafo unscaled scaled}\amsHackForListFollowsLemma
\begin{enumerate}
	\item The kernels $C^q_N$ and $C_N$ of the inverses of the original operator $\mathcal A^q$ and the scaled operator $\mathcal A_N$ are related as follows:
$$	
C_N (x') =L^{\alpha N (d-2)} C^q_N(L^{\alpha N}x')
$$
and thus
\begin{align}  \label{eq:scaling_solution}
(\nabla_k C_N^q \nabla^*_l g^N)(x) =(D_{N,l} C_N D^*_{N,l} g)(L^{-\alpha N} x)
\end{align}
for $g^N(x) = L^{-\alpha N \frac{d}{2}} g(L^{-\alpha N} x)$.
	\item Let $f \in \C(Q(z_1)), g \in \C(Q(z_2))$, $z_1, z_2 \in \R^d$, $0 < \alpha \leq 1$ and define, for $x \in \Lambda_N$,
	$$
	f^N(x) = L^{-\alpha N \frac{d}{2}} f(L^{-\alpha N} x)
	\qquad \text{and} \quad
	g^N(x) = L^{-\alpha N \frac{d}{2}} g(L^{-\alpha N} x).
	$$
Then
$$
\left(\nabla_l^* f^N, \mathcal C^q_N \nabla_l^* g^N\right)
=
\left\langle D^*_{N,l} f, \mathcal C_N D^*_{N,l} g \right\rangle_{l^2}.
$$
\end{enumerate}
\end{lemma}

\begin{proof}
\begin{enumerate}
	\item Fix any discrete function $w \in \chi_N$ and let $u\in \chi_N$ be a unique solution to $\mathcal A^q u (x) = w(x)$ for $x \in \Lambda_N$ and $v'\in \chi_N'$ be a unique solution of $\mathcal A_N v'(x') = w(L^{\alpha N} x')$ for $x' \in \Lambda_N'$. Then
\begin{align*}
w(x)
&= \mathcal A^q u (x)
= \sum_{i,j} a^N_{ij} \, \nabla^*_j \nabla_i u(x)
= L^{-2\alpha N} \sum_{i,j} a^N_{ij} \, D_{N,j}^* D_{N,i} u(L^{\alpha N} x')
\\
&= L^{-2\alpha N} \mathcal A_N u(L^{\alpha N} x')
\end{align*}
and uniqueness of the solution implies $v'(x') = L^{-2\alpha N} u(L^{\alpha N} x')$.
When writing the solutions in terms of the inverse operator kernels $C^q_N$ and $C_N$, we get
\begin{align*}
v' (x')
&= L^{-\alpha Nd} \sum_{y' \in \Lambda'_N} C_N(x'-y') w(L^{\alpha N}y')
\quad \textrm{and}
\\
L^{-2\alpha N} u(L^{\alpha N}x')
& = L^{-2\alpha N} \sum_{y\in \Lambda_N} C^q_N(L^{\alpha N} x' - y) w(y)
\\
& = L^{-\alpha Nd} \sum_{y' \in \Lambda'_N} L^{\alpha Nd} L^{-2\alpha N}C^q_N(L^{\alpha N}(x'-y')) w(L^{\alpha N}y').
\end{align*}
Hence the claim follows.
\item Use the first part of this lemma and insert definitions and scalings.
\end{enumerate}
\end{proof}


\subsection{Convergence of the operators} \label{convergence}

Since the fine-tuning parameter $q^N$ obtained in Proposition \ref{resultAKM} is uniformly bounded by~$1/2$ we get the following convergence result.
	\begin{lemma}\label{convergence of q}
	There exist $\bar{q} \in \R^{d\times d}_{sym}$ and a subsequence $\left(N_k\right)_k$ such that
	$$
	\left\| q^{N_k} - \bar{q} \right\| \rightarrow 0 \text{ as } N \rightarrow \infty.
	$$
	\end{lemma}
	

Define the elliptic differential operator
$$
\mA \defeq \sum_{i,j} a_{ij} \partial_j^* \partial_i
$$
where we use the convention that $\partial_j^* = - \partial_j$. Furthermore let $\div^* \defeq - \div$.
Remark that by the  Lax-Milgram  Theorem there is a (unique up to the addition of constants) solution $u \in W^{1,2}(\T^d)$ such that for $g \in L^2 (\T^d)$  we have $\mA u = \div^* g$ in $\T^d$ and a (unique up to the addition of constants) solution $u \in W^{1,2}_{\loc}(\R^d)$ with $Du \in L^2(\R^d)$ such that for $g \in L^2 (\R^d)$ we have $\mA u = \div^* g$ in $\R^d$.


\begin{prop} \label{convergence torus}
For $f \in L^2(\T^d; \R^d)$ and the corresponding scaled function $f^N$ as defined in \eqref{scaled f} it holds on a subsequence
$$
\left( \sum_l \nabla_l^* f^N_l, \mC^q \sum_l \nabla_l^* f^N_l \right)
\rightarrow \left( f, D \mC \div^* f \right)_{L^2(\T^d;\R^d)}.
$$ 
\end{prop}

\begin{proof}
We first apply Lemma \ref{Trafo unscaled scaled} with $\alpha =1$ to switch to the scaled setting,
$$
\left( \sum_l \nabla_l^* f^N_l, \mC^q \sum_l \nabla_l^* f^N \right)
= \left\langle \div_N^* f, \mC_N \div_N^*f \right\rangle_{l^2}.
$$

Let $u_N = \mC_N \div_N^* f \in \chi_N'$ be the unique solution to $\mA _N u_N = \div_N^* f_N$ (see Proposition~\ref{proposition discrete tori}) with $f_N:= f\vert_{\Lambda_N'}$.

Extend $u_N$ and $f_N$ to piecewise constant functions on the continuous torus $\T^d$.
Divide $\T^d$ into cubes $Q_x$ of side length $L^{-N}$  with centres $x \in \Lambda_N'$ and define $u_N(y) = u_N(x)$ for all $y \in Q_x$. It follows from the definition of the extension (recall also the definition of $\div_N^*$ in \eqref{discrete divergence}) that
$$
\left\langle \div_N^* f_N, u_N \right\rangle_{l^2}
= \left( \div_N^* f_N, u_N \right)_{L^2(\T^d)}
= \left( f_N, D_N u_N \right)_{L^2(\T^d; \R^d)}.
$$
Let $u \in W^{1,2}(\T^d)$ be the solution (unique up to the addition of constants) to $\mA u = \div^* f$ on $\T^d$. From Proposition \ref{proposition discrete tori} it follows that $D_N u_N$ is uniformly bounded in $L^2(\T^d)$. Below in Steps 1-3 we show that on a subsequence $D_N u_N \rightharpoonup D u$ in $L^2(\T^d)$.
Then
\begin{align*}
&\left( \sum_l \nabla_l^* f^N_l, \mC^q \sum_l \nabla_l^* f^N_l \right)
- \left( f, D \mC \div^* f \right)_{L^2(\T^d;\R^d)}\\
&= \left( f_N, D_N u_N \right)_{L^2(\T^d; \R^d)}
- \left( f, D u \right)_{L^2(\T^d;\R^d)} \\
&= \left( f_N - f, D_N u_N \right)_{L^2(\T^d;\R^d)}
+ \left( f, D_N u_N - D u \right)_{L^2(\T^d;\R^d)}\\
&\leq \Vert f_N - f \Vert_{L^2} \Vert D_N u_N \Vert_{L^2} + \left( f, D_N u_N - D u \right)_{L^2}.
\end{align*}
By construction of $f_N$ and by the bound on $D_N u_N$ the first term tends to zero as $N \rightarrow \infty$ and by the weak convergence of $D_N u_N$ this also holds for the second term. Thus the claim follows.


\underline{Step 1:} From the bound on $D_N u_N$ we get existence of $v \in L^2(\T^d;\R^d)$ such that on a subsequence $N_k$
$$
D_{N_k} u_{N_k} \rightharpoonup v \textrm{  in  } L^2(\T^d;\R^d).
$$

\underline{Step 2:}
There is $u \in L^2(\T^d)$ such that $v = Du$ (in the sense of weak derivatives).

\underline{Proof:} Let $\bar{u}_N := \frac{1}{\vert\T^d\vert} \int_{\T^d} u_N \de x$.
We use the discrete Poincar\'{e} inequality and Step 1 to get
\begin{align*}
\Vert u_N - \bar{u}_N \Vert _{L^{2}}
\leq C \Vert D_N u_N \Vert _{L^2}
\leq C.
\end{align*}

Thus there is a subsequence (not denoted explicitly in the following) and $u \in L^{2}(\T^d)$ such that $u_N - \bar{u}_N \rightharpoonup u \,\, \textrm{in} \, L^{2}(\T^d)$.

We take $\varphi \in W^{1,2}(\T^d;\R^d)$ to obtain the following convergence as $N \rightarrow \infty$:
$$
\int D_N u_N \cdot \varphi \,\de x
= \int (u_N - \bar{u}_N) \div_N^* \varphi \,\de x
\rightarrow \int u \div^* \varphi \,\de x
= \int Du \cdot \varphi \,\de x.
$$

On the other hand we have by Step 1 for any $\varphi \in L^2(\T^d; \R^d)$
$$
\int_{\T^d} D_N u_N \varphi \,\de x \rightarrow \int_{\T^d} v \varphi \,\de x
$$
as $N$ tends to infinity.
Thus $u$ is weakly differentiable and $D u = v$.

\underline{Step 3:} The function $u$ in Step 2 satisfies the equation $\mA u = \div^* f$ and thus is unique up to the addition of constants.

\underline{Proof:} For $\varphi \in C^1(\T^d)$ let $\varphi_N$ be the function obtained by restriction to $\Lambda_N'$
 and piecewise constant extension to $\T^d$ and insert $\varphi_N$
 into the weak form of the equation satisfied by $u_N$ to obtain
$$
\sum_{i,j}\int a^N_{i,j} D_{N,i} u_N D_{N,j} \varphi_N \, \de x
= \int f_N \cdot D_N \varphi_N \, \de x.
$$
Now $D_{N,j} \varphi_N$ converges uniformly to $\partial_j \varphi$. Hence the left hand side converges to\linebreak $ \sum_{i,j} \int a_{i,j} \partial_i u \partial_j \varphi \,\de x$  and the right hand side converges to $\int f \cdot D \varphi \, \de x$. Thus
$$ \sum_{i,j} \int a_{i,j} \partial_i u \partial_j \varphi\,\de x = \int f \cdot D \varphi \,\de x$$
for all $\varphi \in C^1(\T^d)$. By density the identity holds for all $\varphi \in W^{1,2}(\T^d)$ and this finishes the proof.

\end{proof}


The proof of the following Proposition is very similar to the proof of Proposition \ref{convergence torus}. 
We just have to take into account that in this case we have to work on increasing tori.

\begin{prop}\label{convergence whole space}
For $J_a, J_b \in L^2(\R^d; \R^d)$ with compact support in $Q(a)$ and $Q(b)$ respectively and the corresponding scaled functions $J_a^N, J_b^N$ as defined in \eqref{scaled J} it holds on a subsequence
$$
\left( \sum_l \nabla_l^* J_a^N, \mC^q \sum_l \nabla_l^* J_b^N \right)
\rightarrow
\left( J_a, D \mC \div^* J_b \right)_{L^2(\R^d;\R^d)}.
$$
\end{prop}

\begin{proof}

First, apply Lemma 4.6 to switch to the scaled setting:
$$
\left( \sum_l \nabla_l^* J_a^N, \mC^q \sum_l \nabla_l^* J_b^N \right)
= \left\langle \div_N^* J_a, \mC_N \div_N^*J_b \right\rangle_{l_2}.
$$

Let $u_N = \mC_N \div_N^* J_b \in \chi_N'$ be the unique solution to $\mA _N u_N = \div_N^* (J_b)_N$ (see Proposition~\ref{proposition discrete tori}), where $(J_b)_N := J_b \vert_{\Lambda_N'}$.

As in the proof of Proposition \ref{convergence torus} we extend $u_N$ and $(J_a)_N, (J_b)_N$ piecewise constant to the continuous torus $\T^d_{R_N}$ with $R_N = L^{(1-\alpha)N}$. It follows from the definition of the extension that
$$
\left\langle \div_N^* (J_a)_N, u_N \right\rangle_{l_2}
= \left( \div_N^* (J_a)_N, u_N \right)_{L_2(\T^d_{R_N})}
= \left( (J_a)_N, D_N u_N \right)_{L_2(\T^d_{R_N};\R^d)}.
$$

Let $u \in W_{\loc}^{1,2}(\R^d)$, $Du \in L^2(\R^d),$ be the solution (unique up to the addition of constants) to $\mA u = \div^* J_b$ on $\R^d$. Let $I_N := \left( - \frac{R_N}{2}; \frac{R_N}{2} \right)$ be the fundamental domain of $\T^d_{R_N}$. We have a uniform bound on $\chi_{I_N} D_N u_N$ in $L^2(\R^d;\R^d)$ by Proposition \ref{proposition discrete tori} and as before we will show that on a subsequence $\chi_{I_N} D_N u_N \rightharpoonup D u$ in $L^2(\R^d;\R^d)$. The claim then follows as in the proof of Proposition \ref{convergence torus}.

\underline{Step 1:} By the uniform bound on $\Vert \chi_{I_{N_k}}D_N u_N \Vert_{L^2(\R^d;\R^d)}$ there is a subsequence $N_k$ and\linebreak $v \in L^2(\R^d;\R^d)$ such that
$$
\chi_{I_{N_k}} D_{N_k} u_{N_k} \rightharpoonup v \,\, \textrm{in} \, L^2(\R^d;\R^d).
$$

\underline{Step 2:} There is $u \in L_{\loc}^2(\R^d)$ such that $v = D u$.

\underline{Proof:} Fix $R>0$ and let $\bar{u}_N := \frac{1}{\vert B_R \vert} \int_{B_R} u_N \de x$. We use the discrete Poincar\'{e} inequality and Step 1 to see
$$
\Vert u_N - \bar{u}_N \Vert _{L^2(B_R)} \leq C(R) \Vert D_N u_N \Vert _{L^2(B_R;\R^d)} \leq C(R).
$$
Thus there is a subsequence $N_k^R$ and $u_{R} \in L^2(B_R)$ such that
$$ 
u_{N_k^R} - \bar{u}_{N_k^R} \rightharpoonup u_{R} \,\, \textrm{in} \, L^2(B_R).
$$
This can be done on arbitrary balls in $\R^d$, and by a diagonal sequence argument (consider the above subsequence $N_k^R$, on $B_{R'}$ it is also bounded, so there is a subsequence and a limit $u_{R'}$ on $B_{R'}$, but on $B_{R'}\cap B_R$ it must hold $u_R = u_{R'}$) there is $u \in L^2_{\loc}(\R^d)$ such that on a subsequence
$$
u_{N} - \bar{u}_{N} \rightharpoonup u \,\, \textrm{in} \, L_{\loc}^2(\R^d).
$$

The rest of the argument is exactly as in the proof of Proposition \ref{convergence torus}.

\underline{Step 3:} $u$ satisfies the equation $\mA u = \div^* f$ and is thus unique up to the addition of constants.

\underline{Proof:} As before (start with a function $\varphi \in C_{\text c}^1(B_R)$ and 
for $L^{(1-\alpha)N} > 2R$ extend $\varphi$ to a function with
period $L^{(1- \alpha)N}$ to deduce the weak form of the limit  equation in $\R^d$  with test function $\varphi$).
\end{proof}


\subsection{Smallness of error terms}\label{errorterms}

Recall the definition of the large field regulator $w_N^{\Lambda_N}$ in
\eqref{large field regulator}.
\begin{lemma} \label{last scale bound}
Let $f \in C^\infty(\T^d;\R^d)$.
For $\xi = 0$ and $\xi = -\mC^q \nabla_l^* f^N$
with $f^N = L^{-N\frac{d}{2}} f(L^{-N} x)$
there is a constant $C$, independent of $N$, such that
$$
w_N^{\Lambda_N} (\xi) \leq C.
$$
\end{lemma}

\begin{proof}
For $\xi = 0$ one computes $w_N^{\Lambda_N} (\xi) =1$ (read carefully the definition of 
the large field regulator, see \eqref{large field regulator}).
For $\xi = -\mC^q \nabla_l^* f^N$ we use \eqref{eq:scaling_solution} in Lemma \ref{Trafo unscaled scaled} for $\alpha=1$ to see that
$$
\nabla^s \mC^q \nabla_l^* f^N(x)
= L^{-N(\frac{d}{2}-1+s)} \mC_N D_N^{s} D_{N,l}^* f(x')  
\quad \mbox{with $x = L^{N} x'$ and $x' \in \T^d$}.
$$
Thus every growing factor $L^N$ in 
$g_{N,x}(\xi)$ and $G_{N,x}(\xi)$ (see \eqref{norm g} and \eqref{norm G}) is perfectly annihilated.
\end{proof}

\begin{lemma}\label{scope alpha}
For $g^N(x) = L^{- \alpha N \frac{d}{2}} g(L^{-\alpha N} x)$, $g \in C^{\infty}_c(\R^d)$, it holds
$$
\left| \mC^q \nabla_l^* g^N \right|_{N, \Lambda_N}
\leq C \, \tau(\alpha)^N
$$
where $C$ is independent of $N$ and $\tau(\alpha) = L^{ (1- \alpha)\left( \frac{d}{2} + 4 \right)}$.
\end{lemma}

\begin{proof}
We use Lemma \ref{Trafo unscaled scaled} to get
\begin{align*}
& \left| \mC^q \nabla_l^* g^N \right|_{N, \Lambda_N}
\\
& = \max_{1\leq s\leq3}
 ~ \sup_{x \in \Lambda_N}
\frac{1}{h} ~ L^{ N\left( \frac{d-2}{2} + s \right)}
\left| \nabla^s \mC^q \nabla_l^* g^N \right|
\\
& \leq \max_{1\leq s\leq3}
 ~ \sup_{x \in \Lambda_N}
\frac{1}{h} ~ L^{N\left( \frac{d-2}{2} + s \right)}
\left|  \sum_y C^q_N(x-y) \left(\nabla^*\right)^{s+1} g^N(y)\right|
\\
& = \max_{1\leq s\leq3}
 ~ \sup_{x \in \Lambda_N}
\frac{1}{h} ~ L^{N\left( \frac{d-2}{2} + s \right)}
\left|  \sum_y C^q_N(x-y) L^{-\alpha N\left( \frac{d}{2} + s + 1 \right)} \left(D_N^*\right)^{s+1} g\left(\frac{y}{L^{\alpha N}}\right)\right|
\\
& = \max_{1\leq s\leq3}
~ \frac{1}{h} ~
L^{N (1- \alpha)\left( \frac{d-2}{2} + s \right)}
~ \sup_{x' \in \Lambda'_N}
\left| \mC_N \left(D_N^*\right)^{s+1} g(x') \right|.
\end{align*}
Apply Corollary \ref{discrete regularity} to see
$$
\sup_{x' \in \Lambda'_N}
\left| \mC_N \left(D_N^*\right)^{s+1} f(x') \right|
\leq C \, L^{2N(1-\alpha)} \left\| g \right\|_{C^{d+s+1}}.
$$
Thus
\begin{align*}
\left| \mC^q \nabla_l^* g^N \right|_{N, \Lambda_N}
\leq C \, \max_{1\leq s\leq3}
~ \frac{1}{h} ~
L^{N (1- \alpha)\left( \frac{d-2}{2} + s \right)}
\, L^{2(1-\alpha)N}
\leq C \, L^{N (1- \alpha)\left( \frac{d}{2} + 4 \right)}.
\end{align*}
Set $\tau(\alpha) = L^{ (1- \alpha)\left( \frac{d}{2} + 4 \right)}$ to obtain the claim.
\end{proof}


\section*{Acknowledgements}
This work was supported by the CRC 1060 {\it The mathematics of emergent effects}. This paper is based on the author's
MSc thesis. I would like to thank my advisor S. M\"uller for his advice and many inspiring discussions.

  \bibliography{meinbib}

\def\cprime{$'$} \def\cprime{$'$}
\begin{thebibliography}{AKM13}

\bibitem[AKM]{AKM15}
S.~Adams, R.~Koteck{\'y}, and S.~M{\"u}ller.
\newblock Strict {C}onvexity of the {S}urface {T}ension for {N}on-convex
  {Potential}.
\newblock {\em in preparation}.

\bibitem[AKM13]{AKM13}
S.~Adams, R.~Koteck{\'y}, and S.~M{\"u}ller.
\newblock Finite range decomposition for families of gradient {G}aussian
  measures.
\newblock {\em J. Funct. Anal.}, 264(1):169--206, 2013.

\bibitem[BBS14]{BBS14}
R.~{Bauerschmidt}, D.~C. {Brydges}, and G.~{Slade}.
\newblock {Critical two-point function of the 4-dimensional weakly
  self-avoiding walk}.
\newblock {\em ArXiv e-prints}, 2014.

\bibitem[BGM04]{BGM04}
D.~C. Brydges, G.~Guadagni, and P.~K. Mitter.
\newblock Finite range decomposition of {G}aussian processes.
\newblock {\em J. Statist. Phys.}, 115(1-2):415--449, 2004.

\bibitem[BK07]{BK07}
M.~Biskup and R.~Koteck{\'y}.
\newblock Phase coexistence of gradient {G}ibbs states.
\newblock {\em Probab. Theory Related Fields}, 139(1-2):1--39, 2007.

\bibitem[Bry09]{Bry09}
D.~C. Brydges.
\newblock Lectures on the renormalisation group.
\newblock In {\em Statistical mechanics}, volume~16 of {\em IAS/Park City Math.
  Ser.}, pages 7--93. Amer. Math. Soc., Providence, RI, 2009.

\bibitem[BS11]{BS11}
M.~Biskup and H.~Spohn.
\newblock Scaling limit for a class of gradient fields with nonconvex
  potentials.
\newblock {\em Ann. Probab.}, 39(1):224--251, 2011.

\bibitem[CD12]{CD12}
C.~Cotar and J.-D. Deuschel.
\newblock Decay of covariances, uniqueness of ergodic component and scaling
  limit for a class of {$\nabla\phi$} systems with non-convex potential.
\newblock {\em Ann. Inst. Henri Poincar\'e Probab. Stat.}, 48(3):819--853,
  2012.

\bibitem[CDM09]{CDM09}
C.~Cotar, J.-D. Deuschel, and S.~M{\"u}ller.
\newblock Strict convexity of the free energy for a class of non-convex
  gradient models.
\newblock {\em Comm. Math. Phys.}, 286(1):359--376, 2009.

\bibitem[DD05]{DD05}
T.~Delmotte and J.-D. Deuschel.
\newblock On estimating the derivatives of symmetric diffusions in stationary
  random environment, with applications to {$\nabla\phi$} interface model.
\newblock {\em Probab. Theory Related Fields}, 133(3):358--390, 2005.

\bibitem[DGI00]{DGI00}
J.-D. Deuschel, G.~Giacomin, and D.~Ioffe.
\newblock Large deviations and concentration properties for {$\nabla\phi$}
  interface models.
\newblock {\em Probab. Theory Related Fields}, 117(1):49--111, 2000.

\bibitem[FS97]{FS97}
T.~Funaki and H.~Spohn.
\newblock Motion by mean curvature from the {G}inzburg-{L}andau {$\nabla \phi$}
  interface model.
\newblock {\em Comm. Math. Phys.}, 185(1):1--36, 1997.

\bibitem[Fun05]{Fun05}
T.~Funaki.
\newblock Stochastic interface models.
\newblock In {\em Lectures on probability theory and statistics}, volume 1869
  of {\em Lecture Notes in Math.}, pages 103--274. Springer, Berlin, 2005.

\bibitem[GOS01]{GOS01}
G.~Giacomin, S.~Olla, and H.~Spohn.
\newblock Equilibrium fluctuations for {$\nabla\phi$} interface model.
\newblock {\em Ann. Probab.}, 29(3):1138--1172, 2001.

\bibitem[KL14]{KL14}
Roman Koteck{\'y} and Stephan Luckhaus.
\newblock Nonlinear elastic free energies and gradient {Y}oung-{G}ibbs
  measures.
\newblock {\em Comm. Math. Phys.}, 326(3):887--917, 2014.

\bibitem[NS97]{NS97}
A.~Naddaf and T.~Spencer.
\newblock On homogenization and scaling limit of some gradient perturbations of
  a massless free field.
\newblock {\em Comm. Math. Phys.}, 183(1):55--84, 1997.

\bibitem[Vel06]{Vel06}
Y.~Velenik.
\newblock Localization and delocalization of random interfaces.
\newblock {\em Probab. Surv.}, 3:112--169, 2006.

\end{thebibliography}
    \bibliographystyle{alpha}

\end{document}